\documentclass[11pt, letter]{article}
\usepackage[margin=1in]{geometry}

\usepackage{amsmath,amsthm, amsfonts}
\usepackage{algpseudocode}
\usepackage{tcolorbox}
\usepackage{bm}
\usepackage{nicefrac}
\usepackage{thm-restate}
\usepackage[boxed, noend]{algorithm2e}
\SetAlCapSkip{1em}

\newtheorem{theorem}{Theorem}[section]
\newtheorem{lemma}[theorem]{Lemma}
\newtheorem{corollary}[theorem]{Corollary}

\theoremstyle{definition}
\newtheorem{definition}[theorem]{Definition}

\theoremstyle{remark}

\usepackage{color}

\newcommand{\etal}{\textit{et al}. }

\def\defeq{\stackrel{\mathrm{def}}{=}}

\newcommand{\indeg}{\boldsymbol{d}_{\textsf{in}}}

\begin{document}
\sloppy

\title{Near-Optimal Fully Dynamic Densest Subgraph}

\date{}

\author{
Saurabh Sawlani\\
Georgia Tech\\
\texttt{sawlani@gatech.edu}
\and
Junxing Wang\\
CMU\\
\texttt{junxingw@cs.cmu.edu}
}
\maketitle

%!TEX root = main.tex

\begin{abstract}
We give the first fully dynamic algorithm which maintains a $(1-\epsilon)$-approximate densest subgraph in worst-case time $\text{poly}(\log n, \epsilon^{-1})$ per update. Dense subgraph discovery is an important primitive for many real-world applications such as community detection, link spam detection, distance query indexing, and computational biology. We approach the densest subgraph problem by framing its dual as a graph orientation problem, which we solve using an augmenting path-like adjustment technique. Our result improves upon the previous best approximation factor of $(\nicefrac{1}{4} - \epsilon)$ for fully dynamic densest subgraph [Bhattacharya \emph{et. al.}, STOC `15].
We also extend our techniques to solving the problem on vertex-weighted graphs
with similar runtimes.

Additionally, we reduce the $(1-\epsilon)$-approximate densest subgraph problem on directed graphs to $O(\log n/\epsilon)$ instances of $(1-\epsilon)$-approximate densest subgraph on vertex-weighted graphs.
This reduction, together with our algorithm for vertex-weighted graphs, gives the first fully-dynamic algorithm for directed densest subgraph in worst-case time $\text{poly}(\log n, \epsilon^{-1})$ per update.
Moreover, combined with a near-linear time algorithm for densest subgraph [Bahmani \emph{et. al.}, WAW `14], this gives the first near-linear time algorithm for directed densest subgraph.
\end{abstract}

\pagenumbering{gobble}

\vfill

\pagebreak

\pagenumbering{arabic}

%!TEX root = main.tex

\section{Introduction} \label{sec:intro}

A majority of real-world networks are very large in size,
and a significant fraction of them are known to change rather rapidly \cite{SahuMSLO17}.
This has necessitated the study of efficient dynamic graph algorithms - 
algorithms which use the existing solution to quickly find
an updated solution for the new graph.
Due to the size of these graphs, it is imperative
that each update be processed in sub-linear time.

Data structures which efficiently maintain solutions to
combinatorial optimization problems have shot into prominence
over the last few decades \cite{SleatorT83, Frederickson85}.
Many fundamental graph problems such as graph connectivity
\cite{HenzingerK99, HolmLT01, KapronKM13},
maximal and maximum matchings \cite{GuptaP13, BhattacharyaHI18, BernsteinS16, BhattacharyaHN16, BhattacharyaHN17},
maximum flows and minimum cuts \cite{ItalianoS10, Thorup07, GoranciHS18}
have been shown to have efficient dynamic algorithms which only require sub-linear
runtime per update.
On the other hand, lower bounds exist for the update times for a number of these problems
\cite{AbboudW14, HenzingerKNS15, AbboudD16, AbboudGIKPTUW19, AbboudKT20}.
\cite{Henzinger18} contains a comprehensive survey of many graph problems and their state-of-the-art dynamic algorithms.

In this paper, we consider the \emph{densest subgraph problem}.
Given an undirected graph $G = \langle V,E \rangle$,
the \emph{density} of a subgraph induced by $S \subseteq V$ 
is defined as $\rho_G(S) = |E(S)|/|S|$,
where $E(S)$ is the set of all edges within $S$.
The densest subgraph problem (DSP) asks to find a a set $S \in V$ such that 
\[
\rho_G^* \defeq \rho_G(S) = \max_{U \subseteq V} \rho_G(U).
\]
We call $\rho_G^*$ the maximum subgraph density of $G$.
%We say that an algorithm outpuis a $(1-\epsilon)$-approximation to DSP
%if it outputs a subset $T \subseteq V$ such that $\rho_G(T) \geq (1-\epsilon)\rho_G^*$.

The densest subgraph problem has great theoretical relevance due to its
close connection to fundamental graph problems such as \emph{network flow}
and \emph{bipartite matching}\footnote{We describe this connection
explicitly in Sections~\ref{sec:prelims} and~\ref{subsec:intuition}.}.
While near-linear time algorithms exist for finding matchings
in graphs \cite{MicaliV80, GabowT91, DuanP14},
the same cannot be said for flows on directed graphs \cite{Madry11}.
In this sense, DSP acts as an indicative middle ground,
since it is both a specific instance of a flow problem
\cite{Goldberg84, BahmaniGM14},
as well as a generalization of bipartite $b$-matchings.
Interestingly, DSP does allow near-linear time
algorithms \cite{BahmaniGM14}.

In terms of dynamic algorithms, the state-of-the-art data structure
for maintaining $(1+\epsilon)$-approximate maximum matchings
takes $O(\sqrt{m} \epsilon^{-2})$ time per update \cite{GuptaP13}.
\cite{BhattacharyaHI18:bmatching} maintain a constant factor approximation
to the $b$-matching problem in $O(\log^3 n)$ time.
For flow-problems, algorithms which
maintain a constant factor approximation in sublinear update time
have proved to be elusive.

In addition to its theoretical importance,
dense subgraph discovery is an important primitive
for several real-world applications such as community detection~\cite{KumarRRT99, Newman06, KumarNT06, DourisboureGP07, ChenS12},
link spam detection~\cite{GibsonKT05}, story identification~\cite{AngelKSSST14},
distance query indexing~\cite{CohenHKZ03, JinXRF09, AkibaIY13} and computational biology \cite{HuYHHZ05, SahaHKRZ10, RenWLW13}, to name a few.
Due to its practical relevance, many related notions of subgraph density,
such as \emph{$k$-cores} \cite{Seidman83}, \emph{quasi-cliques} \cite{BrunatoHB08},
\emph{$\alpha$-$\beta$-communities} \cite{MishraSST08}
have been studied in the literature.
\cite{LeeRJA10, TangL10, Tsourakakis14} contain several other applications of dense subgraphs and related problems.

\subsection{Background and related work}

As defined in \cite{BhattacharyaHNT15},
we say that an algorithm is a \emph{fully dynamic $\gamma$-approximation algorithm}
for the densest subgraph problem if it can process the following operations:
(i) insert/delete an edge into/from the graph;
(ii) query a value which is at least $\gamma$ times the maximum subgraph density of the graph.

Goldberg \cite{Goldberg84} gave the first polynomial-time
algorithm to solve the densest subgraph problem
by reducing it to $O(\log n)$ instances of maximum flow.
This was subsequently improved to use only $O(1)$
instances, using parametric max-flow \cite{GalloGT89}.
Charikar \cite{Charikar00} gave an exact linear programming
formulation of the problem, while at the same time
giving a simple greedy algorithm which gives
a $\nicefrac{1}{2}$-approximate densest subgraph (first studied in \cite{AsahiroITT00}).
Despite the approximation factor,
this algorithm is popular in practice \cite{CohenHKZ03} due to its
simplicity, its efficacy on real-world graphs, and due to the fact that
it runs in linear time and space.

Obtaining fast algorithms for approximation
factors better than $\nicefrac{1}{2}$, however, has proved to be a
harder task.
One approach towards this is to sparsify the graph
in a way that maintains subgraph densities \cite{McGregorTVV15, MitzenmacherPPTX15}
within a factor of $1-\epsilon$,
and run the exact algorithm on the sparsifier.
However, this algorithm still incurs a term
of $n^{1.5}$ in the running time,
causing it to be super-linear
for sparse graphs.
A second approach is via numerical
methods to solve positive LPs\footnote{A positive linear program is one in which all coefficients, variables and constraints are non-negative. They are alternatively known as Mixed Packing and Covering LPs.} approximately.
Bahmani \etal \cite{BahmaniGM14} gave
a $O(m \log n \cdot \epsilon^{-2})$ algorithm
by bounding the width of the dual LP
for this problem,
and using the multiplicative weights update
framework \cite{PlotkinST95, AroraHK12} to find an $(1-\epsilon)$-approximate solution.
Su and Vu \cite{SuV19} used a similar technique to obtain an efficient distributed $(1-\epsilon)$-approximation algorithm.
Alternately, using accelerated methods to solve positive LPs \cite{BoobSW19}
gives a $\widetilde{O}(m \Delta \epsilon^{-1})$ algorithm\footnote{$\widetilde{O}$ hides polylogarithmic factors in $n$.},
where $\Delta$ is the maximum degree in the input graph.

In terms of dynamic and streaming algorithms
for the densest subgraph problem,
the first result is by Bahmani \etal \cite{BahmaniKV12},
where they modified Charikar's greedy algorithm
to give a $(\nicefrac{1}{2} - \epsilon)$-approximation
using $O(\log_{1+\epsilon}n)$ passes over the input.
Das Sarma \etal \cite{SarmaLNT12} adapted this idea to maintain a $(\nicefrac{1}{2} - \epsilon)$ approximate densest subgraph efficiently in the distributed CONGEST model.
Using the same techniques as in the static case,
Bahmani \etal \cite{BahmaniGM14} obtained a $(1-\epsilon)$-approximation
algorithm that requires $O(\log n \epsilon^{-2})$
passes over the input.

Subsequently, Bhattacharya \etal \cite{BhattacharyaHNT15}
developed a more nuanced
data structure to enable a 1-pass streaming algorithm
which finds a $(\nicefrac{1}{2} - \epsilon)$ approximation.
They also gave the first dynamic algorithm
for DSP - a fully dynamic $(\nicefrac{1}{4} - \epsilon)$ approximation algorithm
using amortized time $O(\text{poly}(\log n, \epsilon^{-1}))$
per update.
Around the same time, Epasto \etal \cite{EpastoLS15}
gave a fully dynamic $(\nicefrac{1}{2} - \epsilon)$-approximation algorithm
for DSP in amortized time $O(\log^2 n \epsilon^{-2})$ per update,
with the caveat that edge deletions can only be random.

Kannan and Vinay \cite{KannanV99} defined a notion of density
on directed graphs, and subsequently gave a $O(\log n)$
approximation algorithm for the problem.
Charikar \cite{Charikar00} gave a polynomial-time algorithm
for directed DSP by reducing the problem to solving
$O(n^2)$ LPs. On the other hand, Khuller and Saha \cite{KhullerS09}
used parametrized maximum flow to derive a polynomial-time algorithm.
In the same paper, the gave a linear time 2-approximation algorithm
for the problem.

%Variants of the densest subgraph problem
%have also been studied in the literature.
%One such variant is the \emph{densest at-least-$k$ subgraph
%problem}, which requires the additional constraint
%that the solution have at least $k$ vertices.
%This problem can be shown to be NP-hard \cite{KhullerS09},
%but admits a polynomial-time 2-approximation algorithm \cite{Andersen07, AndersenC09}.
%However, there is evidence that this factor cannot be improved further \cite{Manurangsi18}:
%it is NP-hard to do so assuming the Small Set Expansion Hypothesis\footnote{A computational complexity assumption closely related to the Unique Games Conjecture.}.
%Other variants include \emph{directed densest subgraph} (\cite{KannanV99, Charikar00, KhullerS09}),
%and \emph{clique-densest subgraphs} (\cite{Tsourakakis15, MitzenmacherPPTX15}). 

%An alternate approach towards a dynamic algorithm for the densest subgraph problem
%is to adapt the multiplicative weights update framework \cite{AroraHK12}
%used to solve the densest subgraph problem in \cite{BahmaniGM14}
%to allow for edge updates.
%This technique works in the incremental (only edge insertions) regime for
%bipartite matchings \cite{Gupta14},
%and can similarly be adapted to work in the purely decremental case
%for the densest subgraph problem to give an $O(\log^3 n\epsilon^{-3})$
%amortized runtime per update.

\subsection{Our results}

We use a ``dual" interpretation of the densest subgraph problem
to gain insight on the optimality conditions, as in \cite{Charikar00, BahmaniGM14}.
Specifically, we translate it into a problem of assigning edge loads
to incident vertices so as to minimize the maximum load across vertices.
Viewed another way,
we want to orient edges in a directed graph so as to minimize
the maximum in-degree of the graph.
This view gives a local condition for near-optimality of the algorithm,
which we then leverage to design a data structure
to handle updates efficiently.
As our primary result, we give the first fully dynamic $(1-\epsilon)$-approximation
algorithm for DSP which runs in $O(\text{poly}(\log n, \epsilon^{-1}))$ worst-case time per update:

\begin{restatable}{theorem}{dynamic} \label{thm:dynamic}
	Given a graph $G$ with $n$ vertices,
	there exists a deterministic fully dynamic $(1+\epsilon)$-approximation algorithm
	for the densest subgraph problem using $O(1)$ worst-case time per query
	and 
	$O(\log^4 n \cdot \epsilon^{-6} )$ worst-case time per edge insertion or deletion.

	Moreover, at any point, the algorithm can output the corresponding
	approximate densest subgraph in time $O(\beta + \log n)$,
	where $\beta$ is the number of vertices in the output.
\end{restatable}

Charikar \cite{Charikar00} gave a reduction
from the densest subgraph problem on directed graphs
to solving a number of instances of an LP.
We visualize this LP as DSP on a vertex-weighted graph.
We show that our approach on unweighted graphs extends
naturally to those with vertex weights,
thereby also giving a fully dynamic $(1-\epsilon)$-approximation
algorithm for directed DSP which runs in $O(\text{poly}(\log n, \epsilon^{-1}))$ worst-case time per update:

\begin{restatable}{theorem}{directeddynamic} \label{thm:directeddynamic}
	Given a directed graph $G$ with $n$ vertices,
	there exists a deterministic fully dynamic $(1-\epsilon)$-approximation algorithm for the densest subgraph problem on $G$ using $O(\log n/\epsilon)$ worst-case query time and worst-case update times of
	$O(\log^5 n \cdot \epsilon^{-7} )$ per edge insertion or deletion.
	
	Moreover, at any point, the algorithm can output the corresponding
	approximate densest subgraph in time $O(\beta + \log n)$,
	where $\beta$ is the number of vertices in the output.
\end{restatable}

\subsection{Organization}

In Section~\ref{sec:prelims}, we define essential notation, and formulate DSP as a linear program.
In Section~\ref{sec:dynamic}, we give our primary result - a fully dynamic $1+\epsilon$ approximation algorithm for DSP with
updates in worst-case time $\text{polylog}(n,\epsilon^{-1})$.
In Section~\ref{sec:vertexweighted}, we extend our
results from Section~\ref{sec:dynamic} to vertex-weighted graphs.
In Section~\ref{sec:directed}, we give a detailed reduction from
directed DSP to undirected vertex-weighted DSP.

%!TEX root = main.tex

\section{Preliminaries} \label{sec:prelims}

We represent any undirected graph $G$ as $G = \langle V, E \rangle$,
where $V$ is the set of vertices in $G$,
$E$ is the set of edges in $G$.
For any subset of vertices $S \subseteq V$,
we denote using $E(S)$ the subset of all edges within $S$.

We define $\rho_G(S)$ as the density of subgraph induced by $S$ in $G$, i.e.,
\[
{\rho_G(S)} \defeq \frac{|E(S)|}{|S|}.
\]
The maximum subgraph density of $G$, $\rho^*_G$, is simply
the maximum among all subgraph densities, i.e.,
\[
\rho _G^* \defeq \mathop {\max }\limits_{S \subseteq V} {\rho_G(S)}.
\]

\subsection{LP formulation and dual}
The following is a well-known LP formulation of the densest subgraph problem,
introduced in \cite{Charikar00}.
Associate each vertex $v$ with a variable $x_v \in \{0,1\}$,
where $x_v = 1$ signifies $v$ being included in $S$.
Similarly, for each edge, let $y_e \in \{0,1\}$ denote whether or not it is
in $E(S)$.
Relaxing the variables to be real numbers,
we get the following LP, which we denote by $\textsc{Primal}(G)$,
whose optimal is known to be $\rho_G^*$.
\begin{center}
	\begin{tcolorbox}[width=0.6\linewidth]
		\begin{center}
			$\textsc{Primal}(G)$
		\end{center}
		\vspace*{-5mm}
		\begin{equation*}
		\begin{array}{ll@{}ll}
		\text{maximize}  & \displaystyle\sum\limits_{e \in E} & y_e &\\
		\text{subject to} & & y_e \leq x_u,x_v,        & \forall e=uv \in E\\
		& \displaystyle\sum\limits_{v \in V} & x_v \leq 1,\\
		&                  & y_e \geq 0, x_v \geq 0,         & \forall e \in E, \forall v \in V
		\end{array}
		\end{equation*}
	\end{tcolorbox}
\end{center}

As in \cite{BahmaniGM14, SuV19}, we take greater interest in the dual of the above problem.
Let $f_{e}(u)$ be the dual variable associated with the first $2m$
constraints of the form $y_e \leq x_u$ in $\textsc{Primal}(G)$,
and let $D$ be associated with the last constraint.
We get the following LP, which we denote by $\textsc{Dual}(G)$.
\begin{center}
	\begin{tcolorbox}[width=0.7\linewidth]
		\begin{center}
			$\textsc{Dual}(G)$
		\end{center}
		\vspace*{-5mm}
		\begin{equation*}
		\begin{array}{lr@{}ll}
		\text{minimize}  & & D &\\
		\text{subject to} & f_e(u) + & f_e(v) \geq 1, \qquad & \forall e=uv \in E\\
		&\displaystyle\sum\limits_{e \ni v} & f_e(v) \leq D, & \forall v \in V\\
		&                  & f_e(u) \geq 0, f_e(v) \geq 0,        & \forall e=uv \in E
		\end{array}
		\end{equation*}
	\end{tcolorbox}
\end{center}
This LP can be visualized as follows.
Each edge $e=uv$ has a load of $1$,
which it wants to assign to its end points: $f_e(u)$ and $f_e(v)$
such that the total load on each vertex is at most $D$.
The objective is to find the minimum $D$ for which such a load
assignment is feasible.

For a fixed $D$, the above formulation resembles
a bipartite graph between edges and vertices.
Then, the problem is similar to a bipartite $b$-matching problem \cite{BhattacharyaHI18:bmatching},
where the demands on one side are at most $D$,
and the other side are at least $1$.

From strong duality, we know that the optimal objective values
of both linear programs are equal, i.e., exactly $\rho_G^*$.
Let $\rho_G$ be the objective of any feasible solution to $\textsc{Primal}(G)$.
Similarly, let $\hat \rho_G$ be the objective of any feasible solution to $\textsc{Dual}(G)$.
Then, by optimality of $\rho_G^*$ and weak duality,
\begin{equation} \label{eqn:duality}
\rho_G \leq \rho_G^* \leq \hat\rho_G.
\end{equation}
%!TEX root = main.tex

\section{Fully Dynamic Algorithm} \label{sec:dynamic}

In this section, we describe the main result of the paper: a deterministic fully-dynamic
algorithm which maintains a $(1-\epsilon)$-approximation to the densest subgraph problem in $\text{poly}(\log n, \epsilon^{-1})$ worst-case time per update.

\subsection{Intuition and overview} \label{subsec:intuition}
At a high level, our approach is to view the densest subgraph problem
via its dual problem, i.e., ``assigning" each edge
fractionally to its endpoints (as we discuss in Section~\ref{sec:prelims}).
We view this as a load distribution problem,
where each vertex is assigned some load from its incident edges.
Then, the objective of the problem is simply to find an assignment
such that the maximum vertex load is minimized.
It is easy to verify that an optimal load assignment in the dual problem
is achieved when no edge is able to reassign its load such that the maximum load
among its two endpoints gets reduced.
In other words, local optimality implies global optimality.

In fact, this property holds even for approximately optimal solutions.
We show in Section~\ref{subsec:localopt} that any solution $\bm{f}$ which satisfies an $\eta$-additive approximation to local optimality guarantees an approximate global optimal solution
with a multiplicative error of at most $1 - O(\sqrt{\eta\log n/\hat{\rho}_G})$,
where $\hat{\rho}_G$ denotes the maximum vertex load in $\bm{f}$.
Here, an $\eta$-additive approximation implies that for any edge,
the maximum among its endpoint loads can only be reduced by
at most $\eta$ by reassigning the edge.
So, given an estimate of $\hat{\rho}_G$ and a desired approximation factor $\epsilon$,
we can deduce the required slack parameter $\eta$,
which we will alternatively denote as a function $\eta(\hat{\rho}_G, \epsilon)$.

To do away with fractional edge assignments, in Section~\ref{subsec:orientation}
we scale up the graph by duplicating each edge an appropriate number of times.
When $\eta$ is an integer,
one can always achieve an $\eta$-additive approximation to local optimality
by assigning each edge completely to one of its endpoints.
We visualize such a load assignment via a directed graph,
by orienting each edge towards the vertex to which it is assigned.
Now, the load on every vertex $v$ is simply its in-degree $\indeg(v)$.
Then, an $\eta$-approximate local optimal solution is achieved by orienting each edge such that there is no edge $\overrightarrow{uv}$ with
$\indeg(v) -\indeg(u) > \eta$, because otherwise,
we can flip the edge to achieve a better local solution.
Let us call this a \emph{locally $\eta$-stable oriented graph}.

This leaves the following challenges in extending this idea to a fully dynamic algorithm:
\vspace*{-2mm}
\begin{enumerate}\itemsep=0pt
	\item How can we maintain
	a \emph{locally $\eta$-stable oriented graph} under insertion/deletion operations efficiently?
	\item How do we maintain an accurate estimate of $\eta$ while the graph (and particularly $\hat{\rho}_G$) undergoes changes?
\end{enumerate}  
\vspace*{-2mm}

In Sections~\ref{subsec:datastructure} and~\ref{subsec:threshold}, we solve the first issue
using a technique similar to 
that used by Kopelowitz \etal \cite{KopelowitzKPS14} for the graph orientation problem.
When an edge is inserted or deleted, it causes a vertex
to change its in-degree, which might cause an incident edge to break the
invariant for local $\eta$-stability.
If we flip the edge to fix this instability,
it might cause further instabilities.
To avoid this cascading of unstable edges,
we first identify a maximal chain of ``tight" edges -
edges that are close to breaking the local stability
constraint, and flip all edges in such a chain.
This way, we only increment the degree of the last vertex in the chain.
Since the chain was maximal, this increment maintains the stability condition.
By defining a ``tight" edge appropriately, and applying the same argument to the deletion operation,
we show that each update incurs at most $O(\hat{\rho}_G/\eta)$ flips. 
This chain of tight edges closely relates to the concept of augmenting paths
in network flows \cite{FordF10} and matchings \cite{MicaliV80, DuanP14},
which seems fitting, considering our intuition that densest subgraph
relates closely to these problems.

In Section~\ref{subsec:overallalgo}, we solve the second issue - by simply
running the algorithm for $O(\log n)$ values of $\eta$, and using the appropriate
version of the algorithm to query the solution.

\subsection{Sufficiency of local approximation} \label{subsec:localopt}

From Equation~\ref{eqn:duality}, we know that the optimal solution to $\textsc{Dual}(G)$ gives the
exact maximum subgraph density of $G$, $\rho^*_G$.
Let us interpret the variables of $\textsc{Dual}(G)$ as follows:
\begin{itemize}\itemsep=0pt
	\item Every edge $e=uv$ assigns itself fractionally to one of its two endpoints.
	$f_e(u)$ and $f_e(v)$ denote these fractional loads.
	\item $\sum_{e \ni v} f_e(v)$ is the total load assigned to $v$.
	We denote this using $\ell_v$.
	\item The objective is simply $\max_{v \in V} \ell_v$.
\end{itemize}

If there is any edge $e=uv$ such that $f_e(u) > 0$ and $\ell_u > \ell_v$.
Then $e$ can transfer an infinitesimal amount of load from $u$ to $v$
while not increasing the objective.
Hence, there always exists an optimal solution where for any edge $e=uv$,
$f_e(u) > 0 \implies \ell_u \leq \ell_v$.
Using this intuition, we write the approximate version of $\textsc{Dual}(G)$
by providing a slack of $\eta$ to the above condition. We call this
relaxed LP as $\textsc{Dual}(G, \eta)$.

\begin{center}
	\begin{tcolorbox}[width=0.7\linewidth]
		\begin{center}
			$\textsc{Dual}(G, \eta)$
		\end{center}
		\vspace*{-5mm}
		\begin{equation*}
		\begin{array}{lr@{}ll}
		&  \ell_v & = \displaystyle\sum\limits_{e \ni v}  f_e(v)  \qquad & \forall u \in V \\
		& f_e(u) +  f_e(v) & = 1,                   & \forall e=uv \in E\\
		&                   f_e(u), f_e(v) & \geq 0,        & \forall e=uv \in E\\
		&                 	\ell_u & \leq \ell_v + \eta,        & \forall e=uv \in E, f_e(u) > 0
		\end{array}
		\end{equation*}
	\end{tcolorbox}
\end{center}

Theorem~\ref{thm:lp-approx} states that this local condition
is, in fact, also sufficient to achieve global near-optimality.
Specifically, it shows that $\textsc{Dual}(G, \eta)$
provides a $\nicefrac{1}{(1-\epsilon)}$-approximation to $\rho^*_G$,
where $\eta$ is a parameter depending on $\epsilon$ described later.
Kopelowitz \etal \cite{KopelowitzKPS14} use an identical argument
to show the sufficiency of local optimality for the graph orientation problem.

\begin{theorem} \label{thm:lp-approx}
	Given an undirected graph $G$ with $n$ vertices,
	let $\bm{\hat{f}}, \bm{\hat{\ell}}$ denote any feasible solution to $\textsc{Dual}(G, \eta)$,
	and let $\hat{\rho}_G \defeq \max_{v \in V} \hat{\ell}_v$.
	Then,
	\[
	\left(1 - 3\sqrt {\dfrac{\eta\log n}{\hat{\rho}_G}} \right) \cdot {\hat{\rho}_G}
	\leq
	\rho_G^*
	\leq
	\hat{\rho}_G.
	\]
\end{theorem}

\begin{proof}
	Any feasible solution of $\textsc{Dual}(G, \eta)$ is also a feasible
	solution of ${\textsc{Dual}}(G)$, and so we have $\rho_G^* \leq \hat{\rho}_G$.
%	Also, when $\hat{\rho}_G < 9 \eta\log n$, the first inequality holds by default.
%	It remains to show the first inequality holds for $\hat{\rho}_G \geq 9 \eta \log n$.

	Denote by $T_i$ the set of vertices with load at least $\hat{\rho}_G - \eta i$, i.e.,
	$
	T_i \defeq \left\{ v \in V \mid \hat{\ell}_v \geq \hat{\rho}_G - \eta i \right\}.
	$
	Let $0 < r < 1$ be some adjustable parameter we will fix later.
	We define $k$ to be the maximal integer such that for any $1 \leq i \leq k$,
	$
	|T_i| \geq |T_{i-1}|(1+r).
	$
	Note that such a maximal integer $k$ always exists because there are finite number of vertices in $G$ and the size of $T_i$ grows exponentially. By the maximality of $k$,
	$
	|T_{k+1}| < |T_k|(1+r).
	$
	In order to bound the density of this set $T_{k+1}$, we compute the total
	load on all vertices in $T_{k}$.
	For any $u \in T_k$,
	the load on $u$ is given by
	\[
	\hat{\ell}_u = \sum_{uv \in E} \hat{f}_{uv}(u).
	\]
	However, we know that
	$
	f_{uv}(u) > 0 \implies \hat{\ell}_v \geq \hat{\ell}_u-\eta,
	$
	and hence we only need to count for $v \in T_{k+1}$.
	Summing over all vertices in $T_{k+1}$, we get
	\[
	\sum_{u \in T_k} \hat{\ell}_u = \sum_{u \in T_k, v \in T_{k+1}} \hat{f}_{uv}(u)  \leq \sum_{u \in T_{k+1}, v \in T_{k+1}} \hat{f}_{uv}(u) = |E(T_{k+1})|.
	\]
	Consider the density of set $T_{k+1}$,
	\[
	\rho_G(T_{k+1}) = \dfrac{|E(T_{k+1})|}{|T_{k+1}|} \geq \dfrac{\sum_{u \in T_k} \hat{\ell}_u}{|T_{k+1}|} \geq \dfrac{|T_k|\cdot (\hat{\rho}_G - \eta k)}{|T_{k+1}|},
	\]
	where the last inequality follows from the definition of $T_k$.
	
	Using the fact that $|T_k|/|T_{k+1}| > 1/(1+r) \geq 1-r$,
	\begin{align*}
	\rho_G(T_{k+1}) \geq (1-r)(\hat{\rho}_G - \eta k) \geq \hat{\rho}_G (1-r)\left( 1 - \dfrac{2\eta\log n}{r \cdot \hat{\rho}_G} \right),
	\end{align*}
	where the last inequality comes from the fact that 
	$n \geq |T_{k}| \geq (1+r)^k$, which implies that
	$k \leq \log_{1+r} n \leq 2 \log n / r$.
	
	Now, we can set our parameter $r$ to maximize the term on the RHS.
	By symmetry, the maximum is achieved when both terms in the product are equal
	and hence we set
	\[
	r = \sqrt{\dfrac{2 \eta\log n}{\hat{\rho}_G}}.
	\]
	This gives
	\[
	\rho_G(T_{k+1}) \geq \hat{\rho}_G \cdot \left( 1 - \sqrt{\dfrac{2\eta \log n}{\hat{\rho}_G}} \right)^2 \geq
	\hat{\rho}_G \cdot \left( 1 - 2 \sqrt{\dfrac{2\eta \log n}{\hat{\rho}_G}} \right) \geq
	\hat{\rho}_G \cdot \left( 1 - 3 \sqrt{\dfrac{\eta\log n}{\hat{\rho}_G}} \right).
	\]
Lastly, since $\rho_G(T_{k+1})$ can be at most the maximum subgraph density $\rho_G^*$,
the theorem follows.
\end{proof}

%\begin{remark} \label{rem:subgraph}
The set $T_{k+1}$, in the above proof, is actually a subgraph of $G$
with density at least $\rho_G^*(1-3\sqrt{\eta \log n/\hat\rho_G})$.
%\end{remark}
However, we need the exact value of $\hat{\rho}_G$ to find this set.
As we will see in later sections, we will only have access to an estimate $\rho^{\textsf{est}}$ of the form:
$
\rho^{\textsf{est}} \leq \hat{\rho}_G \leq 2\rho^{\textsf{est}}.
$
So, if we instead set
\begin{align} \label{eqn:r}
r = \sqrt{\dfrac{2\eta\log n}{\rho^{\textsf{est}}}},
\end{align}
we get
%\[
%\sqrt{\dfrac{\eta\log n}{\hat{\rho}_G}} \leq r \leq \sqrt{\dfrac{2\eta\log n}{\hat{\rho}_G}}.
%\]
\[
\rho_G(T_{k+1}) \geq \hat{\rho}_G \cdot \left( 1 - \sqrt{\dfrac{2\eta \log n}{\hat{\rho}_G}} \right)\left( 1 - 2\sqrt{\dfrac{\eta \log n}{\hat{\rho}_G}} \right) \geq
\hat{\rho}_G \cdot \left( 1 - 4 \sqrt{\dfrac{\eta\log n}{\hat{\rho}_G}} \right).
\]
Using the fact that $\hat{\rho}_G \geq \rho_G^*, \rho^{\textsf{est}}$ gives us the following corollary.

\begin{corollary} \label{cor:subgraph}
\[
\rho_G(T_{k+1}) \geq \rho_G^* \cdot \left( 1 - 4 \sqrt{\dfrac{\eta\log n}{\rho^{\textsf{est}}}} \right),
\]
where $T_{k+1}$ is as defined in the proof of Theorem~\ref{thm:lp-approx},
using the value of $r$ as defined in \eqref{eqn:r}.
\end{corollary}
We can now set $\eta$ corresponding to the desired error $\epsilon$
and the estimate $\rho^{\textsf{est}}$.

\subsection{Equivalence to the graph orientation problem}
\label{subsec:orientation}
To obtain an $\epsilon$ approximation,
we need to set $\eta = \dfrac{\epsilon^2 \rho^{\textsf{est}}}{16 \log n}$.
For simpler analysis and to avoid working with fractional loads,
we duplicate each edge $\alpha \defeq \dfrac{64 \log n}{\epsilon^2}$ times.
By doing this, we ensure that $\rho^{\textsf{est}} \geq \hat{\rho}_G/2 \geq \rho_G^*/2 \geq \alpha/4$, 
and thus, $\eta \geq 1$.
This means we can do away with fractional assignments of edges and so each edge $u,v$ is now assigned to either $u$ or $v$.
We can now frame the question as follows:

\begin{tcolorbox}
Given an undirected graph $G$ and an integer $\eta$, we want to assign directions to edges in such a way that for any edge $\overrightarrow{uv}$,
\[
\indeg(v) \leq \indeg(u)+\eta.
\]
\end{tcolorbox}

The above graph orientation problem, i.e., dynamically orienting edges of a graph
to minimize the maximum in-degree, is well studied \cite{BrodalF99,Kowalik07,KopelowitzKPS14}.
Kopelowitz \etal give an efficient dynamic algorithm for the problem,
where the update time depends on the arboricity\footnote{Arboricity is an alternate measure of density defined as $\alpha_G(V) = |E(V)|/(|V|-1)$, and is within $O(1)$ of our density measure.} of the graph with worst-case time bounds.
Our technique for inserting and deleting edges mimics the algorithm by Kopelowitz \etal \cite{KopelowitzKPS14}. However, for our problem, the slack parameter $\eta$ grows linearly with the maximum vertex load.
Hence, we can exploit this additional power to arrive at worst-case times independent of any measure of actual density in the graph.
Additionally, to bound the cost of a vertex informing its updated degree
to its neighbors, we use a lazy round-robin informing technique,
in which not all neighbors are always informed of the latest updates.
We expand on these details in the rest of the section.

\subsection{Data structure for edge flipping in directed graphs}
\label{subsec:datastructure}

At the lowest level, we want to build a data structure that maintains a directed graph undergoing changes.
Ideally, we want each vertex to know its neighbors' labels, so that we can quickly find
any edge violating or exactly satisfying the approximation condition.
We refer to the latter as a \emph{tight} edge.
However, this property is expensive because each vertex could possibly have too many neighbors to inform.
Specifically, each vertex could have up to $\hat{\rho}_G$ in-neighbors and as many as $n-1$ out-neighbors.

We deal with this issue in the following way.
Since a vertex can have $\Omega(n)$ out-neighbors, it does not inform its changes
to its out-neighbors, but only its in-neighbors.
So, any vertex remembers the labels of its out-neighbors.
Hence, it is easy to find a tight outgoing edge; however,
to find a tight incoming edge, we need to query the labels
of all its in-neighbors.
Hence, both the update subroutines and finding a tight incoming edge - 
use as many as $\hat{\rho}_G$ operations.

However, $\hat{\rho}_G$ can also get prohibitively large when the graph sees many insertions, and can reach $\Omega(n)$ (e.g. in a clique).
To tackle this, we relax the requirement for tightness of an edge: we say that an edge
$\overrightarrow{uv}$ is tight if
$
\indeg(v) \geq \indeg(u) + \eta/2.
$
Now, finding a tight edge becomes less strict - importantly it now suffices to update
one's in-neighbors (or query one's in-neighbors) once every $\eta/4$ iterations.
So, in each update, a vertex $v$ only informs $4 \indeg(v)/\eta$ of its neighbors in round-robin fashion. This reduces the number of operations to $O(\alpha)$ per update, as desired.

\begin{algorithm}[!htb]
	\DontPrintSemicolon
	We maintain the following global data structure:\;
	\nlset{$\bullet$} $\textsc{Labels}$: Balanced binary search tree with all labels. We store the max element separately.\; 
	\BlankLine
	Each vertex $u$ maintains the following data structures:\;
	\BlankLine
	\nlset{$\bullet$} $d(u)$: $u$'s label, initialized to $0$.\;
	\nlset{$\bullet$} $\textsc{InNbrs}_u$: List of $u$'s in-neighbors, initialized to $\emptyset$.\;
	\nlset{$\bullet$} $\textsc{OutNbrs}_u$: Max-priority queue of $u$'s out-neighbors indexed using $d_u$, initialized to $\emptyset$.\;
	%\nlset{$\bullet$} $\textsc{idxw}_u$: Index of the last updated in-neighbor of $u$'s (initialized to $0$).\;
	%\nlset{$\bullet$} $\textsc{idxr}_u$: Index of the last checked in-neighbor of $u$'s (initialized to $0$).\;
	\BlankLine
	\begin{minipage}{0.5\textwidth}
		%\SetAlgoLined
		\DontPrintSemicolon
		\SetKwProg{myalg}{Operation}{}{}
		\SetKwFunction{add}{add}
		\myalg{\add{$\overrightarrow{uv}$}}{
			Add $u$ to $\textsc{InNbrs}_v$\;
			Add $v$ to $\textsc{OutNbrs}_u$ with key $d_u(v) \gets d(v)$\;
		}{}
		\BlankLine
		\SetKwFunction{remove}{remove}
		\myalg{\remove{$\overrightarrow{uv}$}}{
			Remove $u$ from $\textsc{InNbrs}_v$\;
			Remove $v$ from $\textsc{OutNbrs}_u$\;}{}
		\BlankLine
		\SetKwFunction{flip}{flip}
		\myalg{\flip{$\overrightarrow{uv}$}}{
			\remove{$\overrightarrow{uv}$}\;
			\add{$\overrightarrow{vu}$}\;}{}
		\BlankLine
		\SetKwFunction{increment}{increment}
		\myalg{\increment{$u$}}{
			$d(u) \gets d(u)+1$\;
			Update $d(u)$ in $\textsc{Labels}$\;
			%Once every $\eta/2$ updates,\;
			\For{$v \in \text{the next } \frac{4\indeg(u)}{\eta} \text{ } \textsc{InNbrs}_u$}{
				Update $d_v(u) \gets d(u)$ in $\textsc{OutNbrs}_v$}
		}{}
		\BlankLine
		\SetKwFunction{decrement}{decrement}
		\myalg{\decrement{$u$}}{
			$d(u) \gets d(u)-1$\;
			Update $d(u)$ in $\textsc{Labels}$\;
			\For{$v \in \text{the next } \frac{4\indeg(u)}{\eta} \text{ } \textsc{InNbrs}_u$}{
				Update $d_v(u) \gets d(u)$ in $\textsc{OutNbrs}_v$}
		}{}
	\end{minipage}%
	\begin{minipage}{0.5\textwidth}
		\DontPrintSemicolon
		\SetKwProg{myalg}{Operation}{}{}
		%\SetKwFunction{add}{add}
		%\SetKwFunction{remove}{remove}
		\SetKwFunction{tightinnbr}{tight\_in\_nbr}
		\myalg{\tightinnbr{$u$}}{
			\For{$v \in \text{the next } \frac{4\indeg(u)}{\eta} \text{ } \textsc{InNbrs}_u$}{
				\lIf{$d(v) \leq d(u)-\eta/2$}{\KwRet{$v$}}
			}
			\KwRet{$\mathsf{null}$}}{}
		%			\ForEach{$v \in \textsc{InNbrs}_u$}%
		%			{$v \gets \textsc{InNbrs}[u].\mathsf{min}$\;
		%			\lIf{$d_u(v) \leq d(u)-\eta/2$}{\KwRet{$v$}}}
		%			\KwRet{$\mathsf{null}$}}{}
		\BlankLine
		\SetKwFunction{tightoutnbr}{tight\_out\_nbr}
		\myalg{\tightoutnbr{$u$}}{
			$t \gets \textsc{OutNbrs}[u].\mathsf{max}$\;
			\lIf{$d_u(v) \geq d(u)+\eta/2$}{\KwRet{$v$}}
			\lElse{\KwRet{$\mathsf{null}$}}}{}
		\BlankLine
		\SetKwFunction{ulabel}{label}
		\myalg{\ulabel{}}{
			\KwRet{$d(u)$}}{}
		\BlankLine
		\SetKwFunction{maxlabel}{max\_label}
		\myalg{\maxlabel{}}{
			\KwRet{$\textsc{Labels}.\mathsf{max}$}}{}
		\BlankLine
		
		\SetKwFunction{maxlabelset}{maximal\_label\_set}
		\SetKwRepeat{Do}{do}{while}
		\myalg{\maxlabelset{$r$}}{
			$m \gets \mathtt{max\_label}()$\;
			\Do{$|B|/|A| \geq 1+r$}{
				$A \gets \text{elements}\geq m-\eta$ \text{ in \textsc{Labels}}\;
				$B \gets \text{elements}\geq m-2\eta$ \text{ in \textsc{Labels}}\;
				$m \gets m - \eta$\;
			}
			\KwRet{$B$}
		}{}		
	\end{minipage}
	\caption{$\textsc{LazyDirectedLabels}(G,\eta)$: A data structure to maintain a directed graph with vertex labels. $V$ and $\eta$ are known.\label{alg:datastructure}}
\end{algorithm}

\begin{lemma}\label{lem:datastructure}
There exists a data structure $\textsc{LazyDirectedLabels}(G,\eta)$ which can maintain a directed graph $G(V,E)$,
appended with vertex labels $d : V \mapsto \mathbb{Z}^+$
while undergoing the following operations:
\begin{itemize}\itemsep=0pt
	\item $\mathtt{add}\left(\overrightarrow{uv}\right)$: add an edge into $G$,
	\item $\mathtt{remove}\left(\overrightarrow{uv}\right)$: remove an edge from $G$,
	\item $\mathtt{increment}(u)$: increment $d(u)$ by $1$,
	\item $\mathtt{decrement}(u)$: decrement $d(u)$ by $1$,
	\item $\mathtt{flip}\left(\overrightarrow{uv}\right)$: flip the direction of an edge in $G$,
%	\item $\mathtt{tight\_in\_nbr}(u)$: find an in-neighbor $v$ with $d(v) \in [d(u) - \eta, d(u) - \eta/2]$, and
	\item $\mathtt{tight\_in\_nbr}(u)$: find an in-neighbor $v$ with $d(v) \leq d(u) - \eta/2$, and
%	\item $\mathtt{tight\_out\_nbr}(u)$: find an out-neighbor $v$ with $d(v) \in [d(u) + \eta/2, d(u) + \eta]$.
	\item $\mathtt{tight\_out\_nbr}(u)$: find an out-neighbor $v$ with $d(v) \geq d(u) + \eta/2$.
	\item $\mathtt{label}(u)$: output $d(u)$.
	\item $\mathtt{max\_label}()$: output $\max_{v \in V}d(v)$.
	\item $\mathtt{maximal\_label\_set}(r)$: Output all elements with labels $\geq \mathtt{max\_label}() - \eta \cdot i$, where $i$ is the smallest integer such that
	$
	|\text{labels} \geq \eta \cdot (i+1)| < (1+r)|\text{labels} \geq \eta \cdot i |.
	$
\end{itemize}
Moreover, the operations $\mathtt{add}$, $\mathtt{remove}$ and $\mathtt{flip}$ can be processed in $O(\log n)$ time;
$\mathtt{tight\_in\_nbr}$, $\mathtt{increment}$ and $\mathtt{decrement}$ can be processed in $O(\alpha)$ time;
and $\mathtt{tight\_out\_nbr}$ and $\mathtt{max\_label}$ can be processed in $O(1)$ time.
$\mathtt{maximal\_label\_set}$ can be processed in time in the order of the output size.

The pseudocode for this data structure is in Algorithm~\ref{alg:datastructure}.
\end{lemma}

\begin{proof}
	The correctness of the data structure follows from the description in Algorithm~\ref{alg:datastructure}.
	The operation $\mathtt{add}$ involves inserting an element into a list and
	a priority queue - giving a worst-case runtime of $O(\log n)$.
	The runtimes for $\mathtt{remove}$ and $\mathtt{flip}$ follow similarly.
	The operations $\mathtt{increment}$ and $\mathtt{decrement}$ involve 1 update to a balanced BST and $O(\alpha)$
	priority-queue updates, giving a worst-case runtime of
	$O(\alpha\log n)$ per call.
	$\mathtt{tight\_in\_nbr}$ queries $O(\alpha)$ neighbors,
	resulting in a worst-case runtime of $O(\alpha)$ per call.
	$\mathtt{tight\_out\_nbr}$, $\mathtt{label}$ and $\mathtt{max\_label}$ simply check an element pointer, resulting in a $O(1)$ runtime.
	Lastly, $\mathtt{maximal\_label\_set}$ traverses a balanced BST, until it exceeds the desired threshold. The time taken is $O(\beta + \log n)$ where $\beta$ is the number of elements read, which is also the size of the output.
\end{proof}

\subsection{Fully dynamic algorithm for a given density estimate}
\label{subsec:threshold}

Here, we assume that an estimate of $\rho_G^*$ (equivalently, an estimate of $\hat{\rho}_G$)
is known. We denote this estimate as $\rho^{\textsf{est}}$, where
$
\rho^{\textsf{est}} \leq \hat{\rho}_G \leq 2\rho^{\textsf{est}}.
$
Using this, we can compute the appropriate $\eta(\rho^{\textsf{est}}, \epsilon) \defeq 2\rho^{\textsf{est}}/\alpha$. Recall that $\alpha$ was defined as
$\alpha \defeq 64 \log n \cdot \epsilon^{-2}$.
From Section~\ref{subsec:datastructure},
we have an efficient data structure to maintain a directed graph,
which we will use to maintain a locally $\eta$-stable orientation.
This in turn gives a fully dynamic algorithm which processes
updates efficiently, as we explain below.

We first define a \emph{tight edge} in a locally stable oriented graph:

\begin{definition}
An edge $\overrightarrow{uv}$ is said to be \emph{tight} if $\indeg(v) \geq \indeg(u) + \eta/2$.
\end{definition}

Now, consider inserting an edge $\overrightarrow{xy}$ into
a locally $\eta$-stable oriented graph.
Since $y$'s in-degree increases, it could potentially
have an in-neighbor $z$ such that $\indeg(z) < \indeg(y)-\eta$.
Note that for this to happen, $\overrightarrow{zy}$ was necessarily a tight edge.
To ``fix" this break in stability, we flip the edge $yz$; however, this causes
$w$'s in-degree to increase, which we now possibly need to fix.
Before explaining how we circumvent this issue,
let us define a \emph{maximal tight chain}.
\begin{definition}
A maximal tight chain \emph{from} a vertex $v$ is a path
of tight edges $\overrightarrow{uv_1}, \overrightarrow{v_1 v_2}, \ldots, \overrightarrow{v_k w}$, such that $w$ has no tight outgoing edges.

A maximal tight chain \emph{to} a vertex $v$ is a path
of tight edges $\overrightarrow{wv_1}, \overrightarrow{v_1 v_2}, \ldots, \overrightarrow{v_k v}$, such that $w$ has no tight incoming edges.
\end{definition}

Now, instead of fixing the ``unstable" edge caused by the increase
in $y$'s in-degree right away,
we instead find a maximal tight chain \emph{to} $y$
and flip all the edges in the chain.
This way, the in-degrees of all vertices in the chain
except the start remain the same.
Due to the maximality of the chain,
the start of the chain has no incoming tight edges,
and hence increasing its in-degree by $1$ will not break local stability.
The same argument holds when we delete $\overrightarrow{xy}$,
except we find a maximal tight chain \emph{from} $y$.

The approximate density is nothing but the highest load in the graph. 
For querying the actual subgraph itself,
we use the observation from Section~\ref{subsec:localopt},
where the required subgraph can be found by:
(i) finding sets of vertices with load at most $\eta \cdot i$ less than the maximum ($T_i$), and
(ii) returning the first $T_{i+1}$ such that $|T_{i+1}|/|T_i| < 1 + r$,
where $r$ is an appropriate function of $\eta$.

\begin{lemma}\label{lem:threshold}
	There exists a data structure $\textsc{Threshold}(G,\eta)$ which can maintain an undirected graph $G(V,E)$
	while undergoing the following operations:
	\begin{itemize}\itemsep=0pt
		\item $\mathtt{insert}(u,v)$: insert an edge into $G$,
		\item $\mathtt{delete}(u,v)$: delete an edge from $G$,
		and report the vertex with decreased load
		\item $\mathtt{query\_load}(u)$: output the current load of $u$.
		\item $\mathtt{query\_density}()$: output a value $\rho_{\mathsf{out}}$ such that $(1-\epsilon)\rho_G^* \leq \rho_{\mathsf{out}} \leq \rho_G^*$.
		\item $\mathtt{query\_subgraph}()$: output a subgraph with density at least $(1-\epsilon)\rho_G^*$.
	\end{itemize}
	Moreover, the operation $\mathtt{insert}$ takes $O(\alpha^2)$ time,
	$\mathtt{delete}$ takes $O(\alpha \log n)$ time,
	$\mathtt{query}$ takes $O(1)$ time,
	and $\mathtt{query\_subgraph}$ takes $O(\beta + \log n)$ time,
	where $\beta$ is the size of the output.

	The pseudocode for this data structure is in Algorithm~\ref{alg:threshold}.
\end{lemma}

\begin{algorithm}[!htb]
\label{alg:threshold}
\DontPrintSemicolon
\nlset{$\bullet$} Initialize data structure $\mathcal{L} \gets \textsc{LazyDirectedLabels}(G,\eta)$
with:\;
$G = (V, \emptyset)$, $\alpha \gets 64\log n \cdot \epsilon^{-2}$, $\eta \gets 2 \rho^{\textsf{est}} /\alpha$\;
%\BlankLine
\begin{minipage}{0.5\textwidth}
\SetKwProg{myalg}{Operation}{}{}
\SetKwFunction{insert}{insert}
\myalg{\insert{$(u,v)$}}{
	\If{$d(u) \geq d(v)$}{
		$\mathcal{L}.\mathtt{add}\left(\overrightarrow{uv}\right)$\;
		$w \gets v$
	}
	\Else{
		$\mathcal{L}.\mathtt{add}\left(\overrightarrow{vu}\right)$\;
		$w \gets u$\;
	}
	\While{$\mathcal{L}.\mathtt{tight\_in\_nbr}(w) \neq \mathsf{null}$}{
		$w' \gets \mathcal{L}.\mathtt{tight\_in\_nbr}(w)$\;
		$\mathcal{L}.\mathtt{flip}\left(\overrightarrow{w'w}\right)$\;
		$w \gets w'$
	}
	$\mathcal{L}.\mathtt{increment}(w)$\;
}{}
\BlankLine
\SetKwFunction{querysubgraph}{query\_subgraph}
\myalg{\querysubgraph{}}{
	$r \gets \sqrt{{2\eta\log n}/{\rho^{\textsf{est}}}}$\;
	\KwRet{$\mathcal{L}.\mathtt{maximal\_label\_set}(r)$};
}{}
\end{minipage}%
\begin{minipage}{0.5\textwidth}
	\SetKwProg{myalg}{Operation}{}{}
	\SetKwFunction{delete}{delete}
	\myalg{\delete{$(u,v)$}}{
		\If{$u \in \textsc{InNbrs}_v$}{
			$\mathcal{L}.\mathtt{remove}\left(\overrightarrow{uv}\right)$\;
			$w \gets v$
		}
		\Else{
			$\mathcal{L}.\mathtt{remove}\left(\overrightarrow{vu}\right)$\;
			$w \gets u$\;
		}
		\While{$\mathcal{L}.\mathtt{tight\_out\_nbr}(w) \neq \mathsf{null}$}{
			$w' \gets \mathcal{L}.\mathtt{tight\_out\_nbr}(w)$\;
			$\mathcal{L}.\mathtt{flip}\left(\overrightarrow{ww'}\right)$\;
			$w \gets w'$
		}
		$\mathcal{L}.\mathtt{decrement}(w)$\;
		$\KwRet(w)$\;
	}{}
	\BlankLine
	\SetKwFunction{querydensity}{query\_density}
	\myalg{\querydensity{}}{
		\KwRet{$\mathcal{L}.\mathtt{max\_label} \times (1-\epsilon)$};
	}{}
	\BlankLine
	\SetKwFunction{queryload}{query\_load}
	\myalg{\queryload{$u$}}{
	\KwRet{$\mathcal{L}.\mathtt{label}(u)$};
	}{}
\end{minipage}
	\caption{$\textsc{Threshold}(G, \rho^{\textsf{est}}, \epsilon)$: Update routines on $G$ when an estimate to its maximum load is known.
	Additionally $V$, $n = |V|$, and $\epsilon$ are known.}
\end{algorithm} 

Let us denote by $d_u(v)$, the apparent label of $v$ as seen by $u$.
This concept is needed because when the label of a vertex changes, it doesn't
relay this change to all its in-neighbors immediately.
However, we can claim the following:

\begin{lemma}\label{lem:invariant}
	The local gap constraint is always maintained, i.e., for any edge $\overrightarrow{uv}$,
	$
	d(v) \leq d(u) + \eta.
	$
\end{lemma}

\begin{proof}
There are two ways that this invariant could become unsatisfied: via a decrement to $u$
or an increment to $v$.

Recall that $v$ informs each in-neighbor its label once every $\eta/4$ updates,
hence $|d_u(v) - d(v)|$ cannot be larger than $\eta/4$.
$u$ only decrements if it cannot find a tight out-neighbor, which means that
$d_u(v) < d(u) + \eta/2$.
Hence, at any instant that $d(u)$ is decremented, $d(v) \leq d(u) + 3\eta/4$.

On the other hand, $d(v)$ is only incremented if $v$ cannot find a tight in-neighbor.
Consider the last time that $d(u)$ is decremented before this instant.
At this point, $d(v) \leq d(u) + 3\eta/4$.
After this, there can only be less than $\eta/4$ increments of $d(v)$ before it queries $d(u)$ and flips.
Hence, $d(v) \leq d(u) + \eta$.
\end{proof}

Using Lemma~\ref{lem:invariant} and Corollary~\ref{cor:subgraph},
we get the following corollary, which shows the correctness of Lemma~\ref{lem:threshold}.

\begin{corollary} \label{cor:correctness_threshold}
Let $\rho_{\mathsf{out}} = (1-\epsilon)\max_{v \in V} d(v)$.
Then, $(1-\epsilon)\rho_G^* \leq \rho_{\mathsf{out}} \leq \rho_G^*$.
\end{corollary}

\begin{proof}[Proof of Lemma~\ref{lem:threshold}]
Corollary~\ref{cor:correctness_threshold} gives the correctness proof.
It remains to show the time bounds.
Note that in both $\mathtt{insert}$ and $\mathtt{delete}$ operations,
the maximum chain of tight edges can only be of length at most $2\hat{\rho}_G/\eta = O(\alpha)$.
The $\mathtt{insert}$ operation calls $\mathtt{add}$ and $\mathtt{increment}$ once, $\mathtt{flip}$ and $\mathtt{tight\_in\_nbr}$  $O(\alpha)$ times. From Lemma~\ref{lem:datastructure},
this results in a worst-case runtime of $O(\alpha^2)$ per insertion.
The $\mathtt{delete}$ operation calls $\mathtt{remove}$ and $\mathtt{decrement}$ once, $\mathtt{flip}$ and $\mathtt{tight\_out\_nbr}$  $O(\alpha)$ times. From Lemma~\ref{lem:datastructure},
this results in a worst-case runtime of $O(\alpha\cdot \log n)$ per deletion.
$\mathtt{query\_density}$ only needs one $\mathtt{max\_label}$ call which is $O(1)$ worst-case.
$\mathtt{query\_load}$ also needs one $\mathtt{label}$ call which is $O(1)$ worst-case.
Lastly, $\mathtt{query\_subgraph}$'s runtime follows from Lemma~\ref{lem:datastructure}.
\end{proof}

\subsection{Overall algorithm}
\label{subsec:overallalgo}

Now, we have a sufficient basis to show our main theorem, which we restate:
\dynamic*

\begin{algorithm}[!htb]
	\label{alg:main}
	\DontPrintSemicolon
	\nlset{$\bullet$} \For{$i \gets 1$ to $\log_2 n$}{
		$\alpha \gets 64 \log n \cdot \epsilon^{-2}$;
		$\rho^{\textsf{est}}_i \gets 2^{i-2} \alpha$\;
		Initialize $\mathcal{T}_i \gets \textsc{Threshold}(G, \rho^{\textsf{est}}_i, \epsilon)$\;
		Initialize a sorted list of edges $\mathtt{pending}_i \gets \emptyset$ using two balanced BSTs (one sorted using the first vertex of the edge, and another using the second)\;
		Set $\mathtt{active} \gets 0$
	}
	\BlankLine
	\BlankLine
	\SetKwProg{myalg}{Operation}{}{}
	\SetKwFunction{query}{query}
	\myalg{\query{}}{
		$\KwRet$ $\mathcal{T}_{\mathtt{active}}.\mathtt{query}()$\;
	}{}
	\BlankLine
	\SetKwFunction{querysubgraph}{query\_subgraph}
	\myalg{\querysubgraph{}}{
		$\KwRet$ $\mathcal{T}_{\mathtt{active}}.\mathtt{query\_subgraph}()$\;
	}{}
	\BlankLine
	\SetKwFunction{insert}{insert}
	\myalg{\insert{$(u,v)$}}{
		\For(\tcp*[f]{duplicating $(u,v)$ $\alpha$ times}){$k \gets 1$ to $\alpha$}{ 
			\lFor(\tcp*[f]{affordable copies}){$i \gets \log_2 n$ to $\mathtt{active}+1$}{
				$\mathcal{T}_{i}.\mathtt{insert}((u,v))$
			}
			$\rho \gets \mathcal{T}_{\mathtt{active}+1}.\mathtt{query}()$\;
			\lIf{$\rho \geq 2\rho^{\textsf{est}}_{\mathtt{active}}$}{
				$\mathtt{active} \gets \mathtt{active}+1$
			}
			\lElse{
				$\mathcal{T}_{\mathtt{active}}.\mathtt{insert}((u,v))$
			}
			
			\For(\tcp*[f]{unaffordable copies}){$i \gets \mathtt{active}-1$ to $1$}{			
				$\ell_u \gets \mathcal{T}_{i}.{\mathtt{query\_load}}(u)$;
				$\ell_v \gets \mathcal{T}_{i}.{\mathtt{query\_load}}(v)$\;
				\lIf{both $\ell_u, \ell_v \geq 2\rho^{\textsf{est}}$}{
					add $(u,v)$ to $\mathtt{pending}_i$
				}
				\lElse(\tcp*[f]{edge is still insertable}){
					$\mathcal{T}_{i}.\mathtt{insert}((u,v))$
				}
			}
		}
	}{}
	\BlankLine
	\SetKwFunction{delete}{delete}
	\myalg{\delete{$(u,v)$}}{
		\For(\tcp*[f]{duplicating $(u,v)$ $\alpha$ times}){$k \gets 1$ to $\alpha$}{ 
			\lFor(\tcp*[f]{affordable copies}){$i \gets \log_2 n$ to $\mathtt{active}+1$}{
				$\mathcal{T}_{i}.\mathtt{delete}((u,v))$
			}
			$\rho \gets \mathcal{T}_{\mathtt{active}}.\mathtt{query}()$\;
			\If{$\rho < \rho^{\textsf{est}}$}{
				$\mathtt{active} \gets \mathtt{active}-1$\;
				%\lFor(\tcp*[f]{process pending edges}){$e \in \mathtt{pending}_{\mathtt{active}}$}{
				%	$\mathcal{T}_{\mathtt{active}}.\mathtt{insert}(e)$
				%}
				%$\mathtt{pending} \gets \emptyset$\;
				$\mathcal{T}_{\mathtt{active}}.\mathtt{delete}((u,v))$\;
			}
			
			\For(\tcp*[f]{unaffordable copies}){$i \gets \mathtt{active}-1$ to $1$}{
				\lIf{$(u,v) \in \mathtt{pending}_i$}{
					remove one copy of $(u,v)$ from $\mathtt{pending}_i$
				}
				\Else{
					$w \gets \mathcal{T}_{i}.\mathtt{delete}((u,v))$
					\tcp*[r]{$w$'s load was decremented}
					\If{$(w,w') \in \mathtt{pending}_i$ for any $w'$}{
						$\mathcal{T}_{i}.\mathtt{insert}((w,w'))$\;
						Remove $(w,w')$ from $\mathtt{pending}_i$\;
					}
				}
			}
		}
	}{}
	\caption{Main update algorithm. $V$, $n=|V|$, and $\epsilon$ are known quantities.}
\end{algorithm} 

From Section~\ref{subsec:threshold},
we now have an efficient fully dynamic data structure $\textsc{Threshold}(G,\rho^{\textsf{est}}, \epsilon)$
to maintain a $1-\epsilon$ approximation to the
maximum subgraph density, provided the optimum remains within
a constant factor of some estimate $\rho^{\textsf{est}}$.
Particularly, $\textsc{Threshold}(G,\rho^{\textsf{est}}, \epsilon)$ requires $\hat{\rho}_G/\eta$ to be small
to work efficiently.
On the other hand, too small an $\eta$ results in a bad approximation factor.

To ensure that we always work with the right estimate $\rho^{\textsf{est}}$,
we will construct $\log_2 n$ copies of $\textsc{Threshold}$,
one copy for each possible $\rho^{\textsf{est}}$, or
equivalently each possible value of $\eta$.
In the $i$th copy of the data structure,
we set $\rho^{\textsf{est}}_i \gets 2^{i-2}\alpha$,
and so $\eta_i \gets 2^{i-1}$.
Let us call this $i$th copy of the data structure as $\mathcal{T}_i \gets \textsc{Threshold}(G,\rho^{\textsf{est}}_i, \epsilon)$.
We also define $\eta_0 = 0$ for the sake of the empty graph.

We say that $\mathcal{T}_i$ is \emph{accurate} if $\rho^{\textsf{est}}_i \leq \hat{\rho}_G$,
or equivalently
$
\eta_i \leq {2\hat{\rho}_G}/{\alpha}.
$
Note that we will never use a copy that is not accurate to deduce the approximate
solution.
On the other hand, we say that $\mathcal{T}_i$ is \emph{affordable} if
the maximum possible chain length is less than $2\alpha$, i.e.,
$
\eta_i > {\hat{\rho}_G}/{\alpha}.
$
On copies that are not affordable, if there are any additions which can cause the maximum load \emph{in that copy} to increase, we hold these off until a later time.

Lastly, note that for any value of $\hat{\rho}_G$, there is exactly one copy which is both accurate and affordable.
We call this the \emph{active} copy. The solution is extracted at any point from this copy.
Suppose the index of the current active copy is $i$. Then, after an insertion,
this can be either $i$ or $i+1$. We first test this by querying the maximum density in $\mathcal{T}_{i+1}$, and accordingly update the active index.
Similarly, after a deletion, this can be $i$ or $i-1$.
For insertions which are not affordable, we store the edges in a \texttt{pending} list.
Consider an insertion $(u,v)$ which is not affordable in $\mathcal{T}_i$.
This means that the loads on both $u$ and $v$ are at the limit $(\eta_i \alpha)$.
We save $(u,v)$ in the pending list.
For $\mathcal{T}_i$ to become affordable, one of $u$'s or $v$'s load must decrease.
At this point, we insert $(u,v)$.
The pseudocode for the overall algorithm is in Algorithm~\ref{alg:main}.

Notice, importantly, that insertions are made into $\mathcal{T}_i$
only when it is affordable. However, we always allow deletions
because these are either deletions from the pending edges
or from the graph currently stored in $\mathcal{T}_i$,
which is still affordable.

\begin{proof}[Proof of Theorem~\ref{thm:dynamic}]
	To show the correctness of Algorithm~\ref{alg:main},
	we need to prove that at all times,
	$
	\hat{\rho}_G/2 \leq \rho^{\textsf{est}}_{\mathtt{active}} < \hat{\rho}_G.
	$ 
	We know that this is true at the start of the algorithm.
	Assume this property is true at some instant before an update.
	When an edge is inserted, the first inequality might break.
	So, we test this after every addition and increment $\mathtt{active}$ accordingly.
	The argument follows similarly for deletions.
	However, we also need to make sure that when some $\mathcal{T}_i$
	is queried, there are no edges remaining in $\mathtt{pending}_i$,
	otherwise the queried density could possibly be incorrect.
	Consider an edge $(u,v)$ inserted into $\mathtt{pending}_i$ at some point
	during the algorithm.
	For $\mathcal{T}_i$ to be queried, it must be active, which means
	that at some point, the load of either $u$ or $v$ decreased,
	causing $(u,v)$ to be inserted. Even when there are multiple such edges
	adjacent to the same high-load vertex, we are assured to see
	at least that many decrements at that vertex.
	
	From Lemma~\ref{lem:threshold}, it follows that a query takes $O(1)$
	worst-case time,
	and finding the subgraph takes $O(\beta + \log n)$ time,
	where $\beta$ is the size of the output subgraph.
	Each insert or delete operation is first duplicated $\alpha$ times.
	Secondly, the updates are made individually in $\log_2 n$ copies of
	the data structure.
	
	First, note that any insert or delete operation
	in $\mathtt{pending}$ can be processed in $O(\log n)$ time.
	This is also true for searching using a single end point of an edge
	owing to the manner in which $\mathtt{pending}$ is defined.
	
	When an edge is added, it makes two load queries
	and then possibly inserts in $\mathcal{T}_i$.
	From Lemma~\ref{lem:threshold},
	this gives a worst-case runtime of $O(\alpha^3\log n)$ time per insertion.
	
	As for deleting an edge,
	it sometimes also requires an insertion into $\mathcal{T}_i$.
	Again, plugging in runtimes from Lemma~\ref{lem:threshold}
	gives a worst-case runtime of $O(\alpha^3\log n)$ time per deletion.
\end{proof}

%!TEX root = main.tex

\section{Vertex-weighted Densest Subgraph} \label{sec:vertexweighted}

In this section, we extend the ideas from Section~\ref{sec:dynamic} to extend
to graphs with vertex weights.
As we will see in Section~\ref{sec:directed},
this extension is crucial in arriving at efficient
dynamic algorithms for DSP on directed graphs.

Let us first formally define the concept
of density in vertex-weighted graphs.
Given a graph $G = \langle V,E,w \rangle$,
where $w : V \mapsto \mathbb{Q}^{\geq 1}$,
the density of a subgraph induced by a vertex subset $S \subseteq V$ is
\[
\rho_G(S) \defeq \dfrac{|E(S)|}{\sum_{v \in S} \omega(v)}.
\]
For ease of notation we denote $\omega(S) \defeq \sum_{v \in S} \omega(v)$.
Constructing the approximate dual like in Sections~\ref{sec:prelims}
and~\ref{sec:dynamic},
we get the same conditions except the \emph{load} on a vertex $v$ is now defined as
\[
\ell_v = \frac{1}{\omega(v)} \sum_{e \ni v} f_e(v).
\]

Let $\omega_{\min}$ and $\omega_{\max}$ denote the smallest and largest vertex weight in $G$.
We multiply all the weights by $1/\omega_{\min}$ and later divide the answer by the same amount. This ensures that all weights are at least $1$, and the maximum
weight is now given by $W \defeq \omega_{\max}/\omega_{\min}$.

We first show that local approximations also suffice for vertex-weighted DSP.
We reuse the notation used in Section~\ref{sec:dynamic} for the exact and approximate dual LP -- $\textsc{Dual}(G)$ and $\textsc{Dual}(G,\eta)$,
but with vertex weights included.

\begin{theorem} \label{thm:lp-approx-vertex-weighted}
	Given an undirected vertex-weighted graph $G$ with $n$ vertices,
	with maximum vertex weight $W$,
	let $\hat{f}, \hat{\ell}$ denote any feasible solution to $\textsc{Dual}(G, \eta)$,
	and let $\hat{\rho}_G \defeq \max_{v \in V} \hat{\ell}_v$.
	Then,
	\[
	\left(1 - 3\sqrt {\dfrac{\eta\log (nW)}{\hat{\rho}_G}} \right) \cdot {\hat{\rho}_G}
	\leq
	\rho_G^*
	\leq
	\hat{\rho}_G.
	\]
\end{theorem}

\begin{proof}
	The proof follows the proof of Theorem~\ref{thm:lp-approx} almost identically.
	
	Any feasible solution of $\textsc{Dual}(G, \eta)$ is also a feasible
	solution of ${\textsc{Dual}}(G)$, and so we have $\rho_G^* \leq \hat{\rho}_G$.
	
	Denote by $T_i$ the set of vertices with load at least $\hat{\rho}_G - \eta i$, i.e.,
	$
	T_i \defeq \left\{ v \in V \mid \hat{\ell}_v \geq \hat{\rho}_G - \eta i \right\}.
	$
	Let $0 < \alpha < 1$ be some adjustable parameter we will fix later.
	We define $k$ to be the maximal integer such that for any $1 \leq i \leq k$,
	$
	\omega(T_i) \geq \omega(T_{i-1}) \cdot (1+\alpha).
	$
	Note that such a maximal integer $k$ always exists because there are finite number of vertices in $G$ and the size of $T_i$ grows exponentially. By the maximality of $k$,
	$
	\omega(T_{k+1}) < \omega(T_k) \cdot (1+\alpha).
	$
	In order to bound the density of this set $T_{k+1}$, we compute the total
	load on all vertices in $T_{k}$.
	For any $u \in T_k$,
	the load on $u$ is given by
	\[
	\hat{\ell}_u = \dfrac{1}{\omega(u)}\sum_{uv \in E} \hat{f}_{uv}(u)
	\]
	However, we know that
	\[
	f_{uv}(u) > 0 \implies \hat{\ell}_v \geq \hat{\ell}_u-\eta
	\]
	and hence we only need to count for $v \in T_{k+1}$.
	Summing over all vertices in $T_{k+1}$, we get
	\[
	\sum_{u \in T_k} \omega(u)\hat{\ell}_u = \sum_{u \in T_k, v \in T_{k+1}} \hat{f}_{uv}(u)  \leq \sum_{u \in T_{k+1}, v \in T_{k+1}} \hat{f}_{uv}(u) = |E(T_{k+1})|.
	\]
	Consider the density of set $T_{k+1}$,
	\[
	\rho(T_{k+1}) = \dfrac{|E(T_{k+1})|}{\omega(T_{k+1})} \geq \dfrac{\sum_{u \in T_k} \hat{\ell}_u}{\omega(T_{k+1})} \geq \dfrac{\omega(T_k)\cdot (\hat{\rho}_G - \eta k)}{\omega(T_{k+1})},
	\]
	where the last inequality follows from the definition of $T_k$.
	
	Since $\rho(T_{k+1})$ can be at most the maximum subgraph density $\rho_G^*$,
	and using the fact that $\omega(T_k)/\omega(T_{k+1}) > 1/(1+\alpha) \geq 1-\alpha$,
	\begin{align*}
	\rho_G^* \geq (1-\alpha)(\hat{\rho}_G - \eta k) \geq \hat{\rho}_G (1-\alpha)\left( 1 - \dfrac{2\eta\log (nW)}{\alpha \cdot \hat{\rho}_G} \right),
	\end{align*}
	where the last inequality comes from the fact that 
	$nW \geq \omega(T_{k}) \geq (1+\alpha)^k$, which implies that
	$k \leq \log_{1+\alpha} (nW) \leq 2 \log (nW) / \alpha$.
	
	Now, we can set our parameter $\alpha$ to maximize the term on the RHS.
	By symmetry, the maximum is achieved when both terms in the product are equal
	and hence we set
	\[
	\alpha = \sqrt{\dfrac{2 \eta\log (nW)}{\hat{\rho}_G}}.
	\]
	This gives
	\[
	\rho_G^* \geq \hat{\rho}_G \cdot \left( 1 - \sqrt{\dfrac{2\eta \log (nW)}{\hat{\rho}_G}} \right)^2 \geq
	\hat{\rho}_G \cdot \left( 1 - 2 \sqrt{\dfrac{2\eta \log (nW)}{\hat{\rho}_G}} \right) \geq
	\hat{\rho}_G \cdot \left( 1 - 3 \sqrt{\dfrac{\eta\log (nW)}{\hat{\rho}_G}} \right). \qedhere
	\]
\end{proof}

Once again, scaling the graph up by a factor of $\alpha \defeq \dfrac{64 \log (nW)}{\epsilon^2}$,
we can frame the question as the following graph orientation problem:

\begin{tcolorbox}
	Given an undirected graph $G$ with vertex-weights $w : V \mapsto \mathbb{Q}^+$ and a slack parameter $\eta$, we want to assign directions to edges in such a way that for any edge $u \rightarrow v$,
	\[
	\dfrac{\indeg(v)}{\omega(v)} \leq \dfrac{\indeg(u)}{\omega(u)}+\eta.
	\]
\end{tcolorbox}

To adapt the data structure from Algorithm~\ref{alg:datastructure},
we only need to make the following change:
\begin{itemize}
	\item $\mathtt{increment}(u)$ and $\mathtt{decrement}(u)$
	no longer increment/decrement by $1$ but by $1/\omega(u)$.
	\item Each entry in the $\textsc{Labels}$ data structure is additionally appended with vertex weights - because instead of computing $|A|$ and $|B|$,
	we need to compute $\omega(A)$ and $\omega(B)$ in $\mathtt{maximal\_label\_set}$.
	\item Since we assumed that $\omega(v) \geq 1$ for all $v \in V$,
	we do not have to adjust the conditions for tight edges.
\end{itemize}

Once we are provided with an estimate of $\hat{\rho}_G$,
we can use the data structure from Algorithm~\ref{alg:threshold}
without any changes.
Similar to Section~\ref{subsec:overallalgo},
we now need to guess a value for $\hat{\rho}_G$.
Notice that the range of values can now be $O(nW)$.
Hence, using $O(\log (nW))$ values,
we can apply Algorithm~\ref{alg:main}
to also solve the vertex-weighted version of DSP.

This gives us the following result.

\begin{theorem} \label{thm:vertexweighted}
	Given a vertex-weighted graph $G$ with $n$ vertices,
	and vertex-weights in the range $\omega_{\min}$ and $\omega_{\max}$,
	there exists a deterministic fully dynamic $(1-\epsilon)$-approximation algorithm
	for the densest subgraph problem on $G$ using $O(1)$ worst-case query time
	and worst-case update times of
	$O(\log^4 (nW) \cdot \epsilon^{-6} )$ per edge insertion or deletion.
	
	Moreover, at any point, the algorithm can output the corresponding
	approximate densest subgraph in time $O(\beta + \log n)$,
	where $\beta$ is the number of vertices in the output.
\end{theorem}
%!TEX root = main.tex

\section{Directed Densest Subgraph} \label{sec:directed}

The directed version of the densest subgraph problem was introduced
by Kannan and Vinay \cite{KannanV99}.
In a directed graph $G = \langle V,E \rangle$, for a pair of sets $S,T \subseteq V$,
we denote using $E(S,T)$ the set of directed edges going from
a vertex in $S$ to a vertex in $T$.
The density of a pair of sets $S,T \subseteq V$
is defined as:
\[
\rho_G(S,T) \defeq \dfrac{|E(S,T)|}{\sqrt{|S||T|}}.
\]
The maximum subgraph density of $G$ is then defined as:
\[
\rho_G^* \defeq \max_{S,T \subseteq V} \rho_G(S,T).
\]
Note that we use the same notation for density for undirected and directed
graphs, as the distinction is clear from the graph in the subscript.

Charikar \cite{Charikar00} reduced directed DSP to $O(n^2)$
instances of solving an LP, and also observed that
only $O(\log n/\epsilon)$ suffice to extract a $(1-\epsilon)$
approximation.
Khuller and Saha \cite{KhullerS09} used the same reduction,
but further simplified the algorithm to 
$O(1)$ instances of a parametrized maximum flow problem.

In this section, we recount this reduction, but by visualizing
the problem reduced to as a densest subgraph problem on
vertex-weighted graphs, as defined in Section~\ref{sec:vertexweighted}.

\subsection{Reduction from Directed DSP to Vertex-weighted Undirected DSP}

Given a directed graph $G = \langle V, E \rangle$ and a parameter $t > 0$, we construct a vertex-weighted undirected graph
\[{G_t} = \langle {V_t},{E_t},{\omega _t}\rangle \]
where,
%\begin{itemize}
%	\item ${V_t} \defeq V_t^{(L)} \cup V_t^{(R)}$, in which $V_t^{(L)}$ and $V_t^{(R)}$ are both clones of the original vertex set $V$;
%	\item ${E_t} \defeq \left\{ {(u,v) \mid u \in V_t^{(L)},v \in V_t^{(R)},(u, v) \in E} \right\}$ projects each original directed edge $(u, v) \in E$ into an undirected edge between $V_t^{(L)}$ and $V_t^{(R)}$, and
%	\item %$\omega_t$ is defined as
%	%\begin{equation*}
%	${\omega _t}(u) \defeq \begin{cases}
%	1 & u \in V_t^{(L)}\\
%	t & u \in V_t^{(R)}
%	\end{cases}$
%	%\end{equation*}
%	
%\end{itemize}
%
\begin{itemize}
	\item ${V_t} \defeq V_t^{(L)} \cup V_t^{(R)}$, in which $V_t^{(L)}$ and $V_t^{(R)}$ are both clones of the original vertex set $V$;
	\item ${E_t} \defeq \left\{ {(u,v) \mid u \in V_t^{(L)},v \in V_t^{(R)},(u, v) \in E} \right\}$ projects each original directed edge $(u, v) \in E$ into an undirected edge between $V_t^{(L)}$ and $V_t^{(R)}$, and
	\item %$\omega_t$ is defined as
	%\begin{equation*}
	${\omega _t}(u) \defeq \begin{cases}
	1/2t & u \in V_t^{(L)}\\
	t/2 & u \in V_t^{(R)}
	\end{cases}$
	%\end{equation*}
	
\end{itemize}

To understand the intuition behind this reduction,
consider a pair of sets $S,T \subseteq V$.
Consider the set $S^{(L)}$ corresponding to $S$ in $V_t^{(L)}$,
and the set $T^{(R)}$ corresponding to $T$ in $V_t^{(R)}$.
$\rho_G(S,T) = \frac{|E(S,T)|}{\sqrt{|S||T|}}$,
whereas $\rho_{G_t}(S^{(L)} \cup T^{(R)}) = \frac{2|E(S,T)|}{(1/t)|S|+t|T|}$.
Picking $t$ carefully lets us relate the two notions,
leveraging the AM-GM inequality as indicated by the two denominators.
Lemmas~\ref{lem:lowerbnd} and~\ref{lem:exact} show this relation in detail.

%\begin{lemma} \label{lem:lowerbnd}
%	For any directed graph $G = \langle V, E \rangle$ and a parameter $t > 0$,
%	\[\hat \rho _G^* \ge 2\sqrt t  \cdot \rho _{{G_t}}^*\]
%\end{lemma}
%
%\begin{proof}
%	
%	Let $G^*_t = \langle V^*_t, E^*_t \rangle$ be the densest (vertex-weighted) subgraph in ${G_t}$. Let $S \defeq V^*_t \cap V_t^{(L)}$ and $T \defeq V^*_t \cap V_t^{(R)}$.
%	Then we have
%	\begin{align*}
%	|E^*_t| &= \rho _{{G_t}}^* \cdot (|S| + t|T|)  \\
%	&\ge 2\sqrt t  \cdot \rho _{{G_t}}^*\sqrt{|S|\cdot|T|},
%	\end{align*}
%	where the inequality follows from the AM-GM property.
%	
%	Now, consider the density of the set $S,T$ in the directed graph $G$.
%	We know that it is at most $\rho^*_G$. 
%	
%	\begin{align*}
%	\rho _G^* &\ge {{\rho }_{S,T}} = \frac{|E(S,T)|}{\sqrt {|S| \cdot |T|}}
%	= \frac{|E^*_t|}{\sqrt {|S| \cdot |T|}} \ge 2\sqrt t  \cdot \rho _{{G_t}}^*. \qedhere
%	\end{align*}
%\end{proof}

\begin{lemma} \label{lem:lowerbnd}
For any directed graph $G = \langle V, E \rangle$,
let $G_t$ be defined as above. Then for any choice of parameter $t$,
\[
\rho _G^* \geq \rho _{{G_t}}^*.
\]
\end{lemma}

\begin{proof}
	
	Let $S^{(L)} \cup T^{(R)}$ denote the densest (vertex-weighted) subgraph in ${G_t}$, where $S^{(L)} \in V_t^{(L)}$ and $T^{(R)} \in V_t^{(R)}$.
	Let $S$ and $T$ denote the corresponding vertex sets in $V$.
	Then we have
	\begin{align*}
	|E_t(S^{(L)} \cup T^{(R)})| &= \rho _{{G_t}}^* \cdot (|S^{(L)}|/t + t|T^{(R)}|) /2 \\
	&\ge \rho _{{G_t}}^*\sqrt{|S^{(L)}|\cdot|T^{(R)}|},
	\end{align*}
	where the inequality follows from the AM-GM property.
	Using the facts $|E_t(S^{(L)} \cup T^{(R)})| = E(S,T)$,
	$|S^{(L)}| = |S|$, and $|T^{(R)}| = |T|$, we get that
	\[
	\dfrac{E(S,T)}{\sqrt{|S|\cdot|T|}} \geq \rho_{G_t}^*.
	\]
	Lastly, since the density of the pair of sets $S,T$ in the directed graph $G$ is at most $\rho^*_G$, we get that $\rho^*_G \geq \rho^*_{G_t}$.
\end{proof}

So, $G_t$ provides a ready lower bound for computing maximum subgraph density,
for any $t$. The next lemma shows that a careful choice of $t$
can give equality between the two optimums.

\begin{lemma} \label{lem:exact}
	For any directed graph $G = \langle V, E \rangle$ and a pair of subsets $S, T$ that provides the maximum subset density, i.e., $\rho _G^* = {{\rho }_G({S,T})}$, we have
	\[
	\rho _G^* = \rho _{{G_t}}^*,
	\]
	where $t = \sqrt{\frac{{| S |}}{{| T |}}}$.
\end{lemma}

\begin{proof}
	Now, consider the sets $S^{(L)} \in V_t^{(L)}$ and $T^{(R)} \in V_t^{(R)}$ corresponding to $S$ and $T$ respectively.
	The density of set $S \cup T$ can be at most $\rho _{{G_t}}^*$:
	\[
	\rho _{{G_t}}^* \ge \frac{{2| {E(S,T)} |}}{{| S |/t + | T | \cdot t}}.
	\]
	Substituting $t$ with $|S|/|T|$,
	\begin{align*}
	\rho _{{G_t}}^* & \geq \frac{2|E(S,T)|}{| S |\sqrt{\frac{| T |}{| S |}} + | T | \sqrt{\frac{| S |}{| T |}}}
	= \frac{{| E(S,T) |}}{{\sqrt {| S | \cdot | T |} }}
	= \hat \rho _G^*.
	\end{align*}
	Combining this with the bound from Lemma \ref{lem:lowerbnd} gives that $\rho _G^* = \rho _{{G_t}}^*$.
\end{proof}

Note, however, that this does not directly give an algorithm for directed densest subgraph, since we do not know the optimum value of $|S|/|T|$.
Since both $|S|$ and $|T|$ are integers between $0$ and $n$,
there can be at most $O(n^2)$ distinct values of $|S|/|T|$.
So, to find the exact solution, we can simply find $\rho_{G_t}^*$
for all possible $t$ values, and report the maximum.

This connection was first observed by Charikar \cite{Charikar00},
where he reduced the directed densest subgraph problem
to solving $O(n^2)$ linear programs.
However, our construction helps view these LPs as
DSP on vertex-weighted graphs, for which there are far more
optimized algorithms than solving generic LPs, in both static
and dynamic paradigms.
Charikar \cite{Charikar00} also observed that a $1+\epsilon$
approximate solution could be obtained by only checking
$O(\log n/\epsilon)$ values of $t$.
As one would expect, to obtain an approximate solution for the directed DSP,
it is not necessary to obtain an exact solution to the undirected vertex-weighted DSP.
As we show in Lemma~\ref{lem:simpleupperbound},
we only require $O(\log n/\epsilon)$ computations
of a $1+\epsilon/2$ approximation to the densest subgraph problem.

\begin{lemma} \label{lem:simpleupperbound}
	For any directed graph $G = \langle V, E \rangle$ and a pair of subsets $S, T$ that provides the maximum subset density, i.e., $\rho_G(S,T) = \rho _G^*$, we have
	\[ \rho _{{G_t}}^*  \geq (1-\epsilon)\rho _G^*,\]
	where $\sqrt{\frac{|S|}{|T|}}\cdot (1-\epsilon) \leq t \leq \sqrt{\frac{|S|}{|T|}}\cdot \frac{1}{(1-\epsilon)}$.
\end{lemma}

\begin{proof}
	Consider the vertices $S^{(L)} \in V_t^{(L)}$ and $T^{(R)} \in V_t^{(R)}$
	corresponding to $S$ and $T$ respectively. The density of set $S^{(L)} \cup T^{(R)}$ can be at most $\rho _{{G_t}}^*$:
	\[
	\rho _{{G_t}}^* \ge \frac{{2| {E(S,T)} |}}{{| S |/t + | T | \cdot t}}.
	\]
	Substituting the bounds for $t$,
	\begin{align*}
	\rho _{{G_t}}^* &\ge \frac{2(1-\epsilon)|E(S,T)|}{|S|\sqrt{\frac{|T|}{|S|}} + | T |\sqrt{\frac{|S|}{|T|}}}
	= (1-\epsilon)\rho_G(S,T)
	= (1-\epsilon) \rho _G^*. \qedhere
	\end{align*}

\end{proof}

\subsection{Implications of the reduction}

The above reduction implies that finding a $(1-\epsilon)$-approximate solution 
to directed DSP can be reduced to $O(\log n/\epsilon)$ instances
of $(1-\epsilon/2)$-approximate vertex-weighted undirected DSP.

\begin{theorem}
	Given a directed graph $G$, with $m$ edges and $n$ vertices, and a $T(m,n,\epsilon)$ time algorithm for $(1-\epsilon)$-approximate vertex-weighted undirected densest subgraph,
	then there exists an $(1-\epsilon)$-approximate algorithm for finding the densest subgraph in $G$ in time $T(m,2n,\epsilon/2) \cdot O(\log n/\epsilon)$.
\end{theorem}

\begin{proof}
	
	For each value of $t$ in
	\[
	\left[\frac{1}{\sqrt{n}}, \frac{1}{(1-\epsilon/2)\sqrt{n}}, \frac{1}{(1-\epsilon/2)^2\sqrt{n}}, \ldots, \sqrt{n} \right],
	\]
	we find an approximate value $\rho$ such that $\rho \geq (1-\epsilon/2)\rho_{G_t}^*$, and output the maximum such value.
	Using $\epsilon/2$ as the error parameter in Lemma~\ref{lem:simpleupperbound},
	we get that $\rho \geq (1-\epsilon)\rho_G^*$.
	
	The number of values of $t$ is $\log_{1/(1-\epsilon/2)} n = O(\log n/\epsilon)$.
\end{proof}

The current fastest algorithms for $(1-\epsilon)$-approximate static densest subgraph \cite{BahmaniGM14, BoobSW19}
rely on approximately solving $\textsc{Dual}(G)$, which is a positive linear program,
and subsequently extracting a primal solution.
Both these parts of the algorithm extend naturally to vertex-weighted graphs.
Substituting these runtimes in for $T(m,n,\epsilon)$, we get the following corollary.

\begin{corollary}\label{cor:staticdirected}
Let $G$ be a directed graph with $m$ edges and $n$ vertices, and let $\Delta$ be the maximum value among all its in-degrees and out-degrees.
Then, there exists an algorithm to find a $(1-\epsilon)$-approximate densest subgraph in $G$ in time $\widetilde{O}(m \epsilon^{-2} \cdot \min(\Delta, \epsilon^{-1}))$.
\end{corollary}
Here, $\widetilde{O}$ hides polylogarithmic factors in $n$.

The same reduction also applies to fully dynamic algorithm for directed DSP.

\begin{theorem}
	Suppose there exists a fully dynamic $(1-\epsilon)$-approximation algorithm for undirected vertex-weighted DSP on an $n$-vertex graph with update time $U(n,\epsilon)$ and query time $Q(n,\epsilon)$.
	Then,
	there exists a deterministic fully dynamic $(1-\epsilon)$-approximation algorithm
	for directed DSP on an $n$-vertex graph using $U(2n,\epsilon/2) \cdot O(\log n /\epsilon)$ query time and $Q(2n,\epsilon/2) \cdot O(\log n /\epsilon)$ query time.
\end{theorem}

Substituting the runtimes from Theorem~\ref{thm:vertexweighted} in Section~\ref{sec:vertexweighted}, we get our result for dynamic DSP on directed graphs.

\directeddynamic*
%\begin{corollary}\label{cor:dynamicdirected}
%	Let $G$ be a directed graph with $n$ vertices.
%	Then, there exists a deterministic fully dynamic $(1-\epsilon)$-approximation algorithm for the densest subgraph problem in $G$, with worst-case query time $O(\log n/\epsilon)$ and worst-case update time $O(\log^5 n \cdot \epsilon^{-7})$.
%\end{corollary}

\section*{Acknowledgements}
We thank Richard Peng and Gary Miller for their feedback and insightful discussions.

{\small
\bibliographystyle{alpha}
\bibliography{DSP_references}
}

%\begin{appendix}
%\end{appendix}

\end{document}